\documentclass[reprint,aps,pra,twocolumn,showpacs]{revtex4-1}
\usepackage{mathtools}
\usepackage{amsmath,amsfonts,amssymb,amsthm,mathrsfs}
\usepackage{booktabs}
\usepackage[T1, OT1]{fontenc}
\DeclareTextSymbolDefault{\DH}{T1}
\usepackage{hyperref}
\usepackage[capitalize]{cleveref}
\hypersetup{ unicode=true, 
  pdftoolbar=true, 
  pdfmenubar=true, 
  pdffitwindow=false, 
  pdfstartview={FitH}, 
  pdftitle={My title}, 
  pdfauthor={Author}, 
  pdfsubject={Subject}, 
  pdfcreator={Creator}, 
  pdfkeywords={keyword1, key2, key3}, 
  pdfnewwindow=true, 
  colorlinks=true, 
  linkcolor=blue, 
  citecolor=blue, 
  filecolor=blue, 
  urlcolor=blue, 
  pdfencoding=auto, psdextra }

\newtheorem{defi}{Definition}

\newtheorem{thm}{Theorem}
\newtheorem{corollary}[thm]{Corollary}

\newtheorem{lem}[thm]{Lemma}

\newtheorem{example}{Example}

\usepackage{tabstackengine}
\stackMath%

 \newcommand{\complex}{\mathbb{C}} 
\newcommand{\rank}[1]{\mathrm{rank} (#1)} \newcommand{\range}[1]{\mathcal{R} (#1)}

 \usepackage{mathtools} 

\def\symbtype{1} \ifcase\symbtype%
\newcommand{\adj}[1]{{#1}^{*}}
\newcommand{\conj}[1]{\overline{#1}}
\DeclarePairedDelimiterX\bra[1]{}{}{\adj{#1}}
\DeclarePairedDelimiter\ket{}{}
\DeclarePairedDelimiterX\braket[2]{}{}{(#1, #2)}
\DeclarePairedDelimiterX\mket[2]{}{}{#1\otimes#2}

\or%
\newcommand{\adj}[1]{{#1}^{\dagger}}
\newcommand{\conj}[1]{{#1}^{*}}
\DeclarePairedDelimiter\bra{\langle}{\rvert}
\DeclarePairedDelimiter\ket{\lvert}{\rangle}
\DeclarePairedDelimiterX\mket[2]{\lvert}{\rangle}{#1,#2}
\DeclarePairedDelimiterX\braket[2]{\langle}{\rangle}{#1 \delimsize\vert#2}

\fi

\newcommand{\sijkl}[2]{S^1_{#1,#2}}
\newcommand{\sijklt}[2]{S^2_{#1,#2}}
\newcommand{\commutes}{\; \substack{\mathrm{commutes}\\
    \xrightarrow{\hspace*{6em}}\\ \mbox{}}\; }
\newcommand\lijkl[2]{\lambda^{1}_{#1 #2}}
\newcommand\lijklt[2]{\lambda^{2}_{#1 #2}}
\newcommand\Lijkl[1]{\Lambda^{1}_{#1 }}
\newcommand\Lijklt[1]{\Lambda^{2}_{#1
  }}
\makeatletter
\newcommand{\notprop}{\propto\kern-1\@ptsize pt \diagup}
\makeatother

\renewcommand{\intercal}{\mathsf{T}}
\renewcommand\tilde[1]{\stackrel{\sim}{\smash{{#1}}\rule{0pt}{1.1ex}}}
\usepackage{dsfont}
\newcommand{\I}{\mathds 1}

 \preprint{APS/123-QED}

\begin{document}
\title{Separability of Multipartite Quantum States with Strong Positive Partial Transpose } \date{}

\begin{abstract}
  We generalize the definition of strong positive partial transpose (SPPT) to the multipartite system. The tripartite
  case was first considered by X.-Y. Yu and H. Zhao [ Int. J. Theor. Phys.,
  \textbf{54}, 292, (2015)]. In this extension, unfortunately, desired properties
  such as the PPT of SPPT states and  the separability of super  and pure SPPT  states  are not preserved.
  In contrast, this paper provides an alternative  generalization to multipartite cases  with these properties
  preserved. 
    We also provide sufficient conditions for the  separability of SPPT states.
\end{abstract}
\author{Lilong Qian}
 \pacs{03.65.Ud, 03.67.Mn}
 \email{Electronic adress: qian.lilong@u.nus.edu}
  \affiliation{Department of Mathematics, National University of Singapore, 119076, Singapore}
\date{\today}
\maketitle

\section{Introduction}
Entanglement lies at the heart of quantum computation and information theory, which is the resource of most
applications in quantum information processing tasks.  Since 1935, when the necessarily nonlocal nature of quantum mechanics was
first highlighted by Einstein, Podolsky and Rosen (EPR)~\cite{einstein1935can}, quantum entanglement has become a major
quantum phenomenon  which requires further understanding. One of the fundamental problems about quantum entanglement is the
 separability problem i.e.\@ to check whether a given quantum state is separable or entangled.  Given a density
matrix $\rho$ in a quantum bipartite system $A:B$,  $\rho$ is  said to be separable if it can be written as a convex combination
of product states~\cite{Werner1989},  i.e. $\rho$ is separable if
\begin{align}
  \label{intro}
  \rho=\sum_i p_i\rho_i^A\otimes\rho_i^B,\;\sum_i p_1 = 1,p_i\geqslant 0,
\end{align}
where $\rho_i^A$ and $\rho^B_i$ are the density matrices in subsystems A and B respectively. A quantum state is said to be
quantum entangled if the density matrix does not possess a decomposition of the form as \cref{intro}. Unlike
separable states, entangled states cannot be obtained by preparing their subsystems~\cite{Peres1996a}.

Despite  remarkable efforts over recent years, the operational necessary and sufficient condition for the
separability  still remains unknown in general. It has been found that the separable problem is NP-HARD even for
the bipartite system~\cite{gurvits2003classical}.

While it is hard to solve this problem in general, there are plenty of  practical criteria which enable us to detect
entanglement for some sub-classes.  One of the most famous criteria  is called positive
partial transpose (PPT)  or Peres-Horodecki criterion~\cite{Peres1996a}.  It tells  that if a state
$\rho$ is separable, then its partial transposed state $\rho^{\intercal_{A}} = (\intercal \otimes\I)\rho$ remains
positive.
Using positive
maps, Horodecki {\sl et al.}~\cite{horodecki1996separability} showed that Peres-Horodecki criterion is also sufficient
for $2\otimes 2$ and $2\otimes 3$  systems. It is, however, not true for higher dimensional spaces.
Woronowicz~\cite{Woronowikz1976} constructed a counterexample of a $2\otimes 4$ entangled PPT state. See more entangle
PPT states in Refs.~\cite{Choi1982,Strmer1963,Horodecki1997}.  Utilizing matrix analysis, Kraus {\sl et \
  al.}~\cite{Kraus2000} showed that any $M\otimes N$ PPT state of rank $N$ is separable. Moreover, some generalized
results  are proposed in Refs.~\cite{Horodecki2000,Karnas2001,Fei2003}. Since any $M\otimes N$ state of rank less than
$N$ is distillable~\cite{horodecki2003rank}, it suffices to consider these state whose rank is greater than its local
ranks.

A subclass of PPT states, namely strong PPT (SPPT) states, were first considered by  {Chru\'sci\'nski} {\sl
  et al.}~\cite{Chruscinski2008}. These states have a ``strong PPT property''. Based on several examples, it was
conjectured that SPPT states are separable. Unfortunately, this conclusion fails for  $M\otimes N$ PPT states when
$NM\geq 9$. Actually, all $2\otimes 4$ SPPT states are separable~\cite{Ha2012}. But, there exists a $2\otimes 5$ SPPT
which is entangled~\cite{Ha2012}. The separability of SPPT states become more complex in high dimensional spaces. The
SPPT states encompass many previously known separable PPT states such as rank $N$ states of $2\otimes N$
system. Moreover, it is proved that SPPT states can be used to witness quantum discord (QD) in $2\otimes N$
systems~\cite{Bylicka2010}.  In addition, Bylicka {\sl et al.}~\cite{Bylicka2013} constructed  a special class of SPPT
states, which were called super strong SPPT (SSPPT) states.  In Ref.~\cite{Guo2012}, the
decomposition of SSPPT states was considered in both finite and infinite dimensional systems.

In a recent paper\cite{SPPT3partite}, the idea of SPPT states was generalized to the tripartite system
$A_{1}:A_{2}:A_{3}$. However, these states are essentially  bipartite SPPT  with respect to the bi-partition
$A_{1}A_{2}:A_{3}$. As a result, some good properties may be lost in the tripartite sense. For instance, 
the SPPT cannot guarantee PPT, which is one of the
most important features for SPPT states in the bipartite system.
Also, the super SPPT cannot  guarantee the separability in general.
Therefore, it would be especially interesting to find a more appropriate generalization to tripartite or even
multipartite systems.  The purpose of this paper is to provide an alternative definition of SPPT states in the $n$-particle
system. We begin with the simplest  $2\otimes 2 \otimes N $ case,  and eventually, extend  to general many-body
system.

The remainder of this paper is organized as follows.  In Section~\ref{sec:preliminaries}, we present some preliminaries
about the separability problem of the SPPT states. In Section~\ref{sec:old-version-sppt}, we recall the definition of
tripartite SPPT states in Ref.~\cite{SPPT3partite}. We showed that the defined SPPT and SSPPT cannot inherit many good
properties as those in the bipartite systems.  In Section~\ref{sec:mySPPT}, we provide a new idea to define the SPPT and
SSPPT state in $2\otimes 2 \otimes N$ system.  We extend this concept to
$N_1\otimes N_2\otimes N_3$ case. Finally, we show  the idea  to the arbitrary multipartite system.
In Section~\ref{sec:suff-cond-sppt}, we  propose
some  sufficient conditions for separability of SPPT states. Some concluding remarks are given in \Cref{conclusion}.

\section{preliminaries}
\label{sec:preliminaries}
We start this section with a formal definition of separability.  Consider a $d$-particle state belonging to a Hilbert
space $\mathcal{H}$. Denote by $A_1$, $A_2$, $\cdots$, $A_d$ the subsystems respectively.  Each subsystem is a Hilbert
space $\mathcal{H}_i$ with dimension $N_i$. By the postulate for composition of system in quantum computation thery, we have
$\mathcal{H} = \otimes_{i=1}^d \mathcal{H}_i$ and $\dim(\mathcal{H})=\prod_{i=1}^{d}N_i$.

To make more concise, we use vector based indexes in this paper. Let $\mathcal{I}$ be the set consisting of d-tuples
$ (i_1,i_2,\ldots,i_d),1\leqslant i_k\leqslant N_k$. Given a index $\alpha= (i_1,i_2,\ldots,i_d)\in \mathcal{I}$, then
$x_{\alpha}$ represents $x_{i_1,i_2,\ldots,i_{d}}$. Furthermore, we also assign an order for those indexes. Let
\[
  n(\alpha)= \prod_{k=1}^{d-1}(\prod_{l=k+1}^{d}N_k)(\alpha_k-1)+\alpha_d,
\]
for all $\alpha\in\mathcal{I}$. Then we say $\alpha\leqslant \beta$ if $n(\alpha)\leqslant n(\beta)$ for
$\alpha,\beta\in \mathcal{I}$.  Moreover, $\ket{\alpha} $ represents the product vector $\ket{i_1,i_2,\ldots,i_d}$.

Hence a density matrix acting on the space $\mathcal{H}$ can be represented as
\begin{equation}
  \label{eq:37}
  \rho = \sum_{\alpha,\beta\leqslant \alpha_0}  \rho_{\alpha,\beta}| \alpha\rangle \langle \beta|, 
\end{equation}
where $\alpha,\beta\in\mathcal{I}$, $\alpha_0 = (N_1,N_2,\ldots,N_d)$.

Now we recall the definition of separability of a quantum state. A density matrix $\rho$ in $\mathcal{H}$ is said to be
separable if it can be written as
\begin{equation}
  \label{eq:12}
  \rho = \sum_{i=1}^{L} \lambda_i\ket{x_i}\bra{x_i},
\end{equation}
where $\sum_{i}^{L}\lambda_i = 1,\lambda_i\geqslant 0$ and each $\ket{x_{i}}$ is a pure product vector in the space
$\mathcal{H}$.

Peres-Horodecki criterion plays a crucial role in the separability problem, which is based on the partial transpose.
Therefore it would be necessarily to introduce the notations of partial transposes ahead of time. Let $\rho$ be a given
state in the composite system $A_1:A_2$. Denote by $\intercal$ the usual transpose operator.  Then the composite
operators $(\I\otimes \intercal)$ and $(\intercal\otimes \I)$ are called the partial transpose operators. Furthermore,
the partial transposed density matrices are denoted by $ \rho^{\intercal_2}=(\I\otimes \intercal)\rho$ and
$\rho^{\intercal_1}=(\intercal\otimes \I)\rho $.  For general $d$-particle system $A_1:A_2\cdots:A_d$, we denote by
$\intercal_i$ ( $i=1,2,\ldots,d$) the partial transpose with respect to $i$-th subsystem respectively. The corresponding
partial transposed state is denoted by $\rho^{\intercal_i}$. Generally, given an index set $I = \{i_1,i_2,\ldots,i_k\}$,
$\intercal_I$ denotes the partial transpose with respect to the subsystems in $I$, that is
\[ \intercal_{I}= \circ_{k\in I}\intercal_k.\] In the $d$-body system, we introduce a special partial transpose,
\begin{equation}
  \label{eq:54}
  \Gamma_k =   \underbrace{\intercal\otimes \intercal\otimes \cdots \otimes \intercal}_{k \intercal}
  \otimes
  \underbrace{\I\otimes \cdots \otimes \I}_{(d-k) \intercal},
\end{equation}
which will be used in the following sections.

PPT  criterion tells  that if $\rho$ is separable, then
\begin{equation}
  \label{multiSPPTv0.6:eq:1}
  \intercal_{\!I} \cdot \rho\geqslant 0,
\end{equation}
for any $I\subset \{1,2,\ldots,d\} $.

Now we recall the definition of SPPT in the bipartite system.

Consider a density matrix $\rho$ in $N_1\otimes N_2$ system with a block Cholesky decomposition $\rho = \adj X X$,
\begin{equation}
  \label{eq:40}
  \begin{split}
    X& =
    \begin{pmatrix}
      X_{11}&\cdots&X_{1N_1}\\
      \vdots& \ddots &\vdots\\
      X_{N_1 1}&\cdots &X_{N_1 N_1} \\
    \end{pmatrix}\\
    &=
    \begin{pmatrix}
      S_{11}X_1 &\cdots&S_{1N_1}X_1\\
      \vdots&\ddots& \vdots\\
      S_{N_1 1}X_{N_1}&\cdots&S_{N_1N_1}X_{N_1}\\
    \end{pmatrix}\\
    & =
    \begin{pmatrix}
      X_1 &\cdots&S_{1N_1}X_1\\
      \vdots&\ddots& \vdots\\
      0&\cdots&X_{N_1}\\
    \end{pmatrix},
  \end{split}
\end{equation}
where $S_{ij}$ and $X_i$ are both $N_1\times N_1$ matrix with
\begin{equation*}
  S_{ij} =
  \begin{cases}
    \I_{N_2}, & \text {if } i=j;\\
    0, & \text{ if }i > j.
  \end{cases}
\end{equation*}
In this paper, $\I_n$ denotes the identity operator acting on the space $\complex^n$.
\begin{defi}
  Let $\rho$ be a density matrix in $N_1\otimes N_2$ system. And $\rho = \adj X X$, where $X$ has the form as
  \cref{eq:40}. Then $\rho$ is said to be SPPT if
  \begin{equation}
    \label{eq:41}
    \rho^{\intercal_1} = \adj Y Y,
  \end{equation}
  with
  \begin{equation*}
    Y =
    \begin{pmatrix}
      &&\\
      &\adj S_{ij}X_i&\\
      &&\\
    \end{pmatrix},
  \end{equation*}
  or equivalently,
  \begin{equation}
    \label{eq:42}
    \sum_{k=1} ^{N_1} \adj X_k [\adj S_{kj},S_{ki}]X_k = 0,\;1\leqslant i \leqslant j\leqslant N_1.
  \end{equation}
\end{defi}

Here the commutator  $[A,B]$ is defined by $[A,B]= AB-BA$.

In particular, \cref{eq:42} is naturally satisfied if
\begin{equation}
  \label{eq:43}
  [\adj S_{kj},S_{ki}]=0,1\leqslant i \leqslant j\leqslant N_1.
\end{equation}
We call this subclass of SPPT states super SPPT (SSPPT) states~\cite{Chruscinski2008}. It was proved that every SSPPT
state is separable~\cite{Bylicka2013}.

\section{Previous definition of tripartite SPPT states}%
~\label{sec:old-version-sppt}
In this section, we will first introduce the definition of tripartite SPPT states in Ref.~\cite{SPPT3partite}.
After that, we will show that the SPPT states will not preserve some good properties as that in the bipartite system. For
example, the SPPT state may not be PPT.\@ Besides, pure or super SPPT states may not be separable.

Suppose that $\rho$ is a density matrix in the tripartite system  $A_1:A_2:A_3$, with a
decomposition $\rho = \adj X X$. Under the bi-partition $A_1:A_2A_3$, $X$ can be written as an $N_1\times N_1$ block matrix:
\begin{equation}
  \label{eq:44}
  X =
  \begin{pmatrix}
    &&\\
    &Z_{ij}\\
    &&\\
  \end{pmatrix},\;Z_{ij}\in M_{N_2N_3}.
\end{equation}
Again, each $Z_{ij}$ can be  written as a $N_2\times N_2$ block matrix:
\begin{equation}
  \label{eq:45}
  Z_{ij} =
  \begin{pmatrix}
    & &\\
    &Z_{ij,kl}&\\
    & &\\
  \end{pmatrix}
  =
  \begin{pmatrix}
    &&\\
    &S_{ijkl}X_{ik}&\\
    &&
  \end{pmatrix},
\end{equation}
where
\begin{equation}
  \label{eq:46}
  S_{ijkl} =
  \begin{cases}
    \I_{N_3}, & (i,k) = (j,l);\\
    0,& (i,k) > (j,l).
  \end{cases}
\end{equation}
Note that the order $(i,k)\leqslant (j,l)$ means $(i-1)N_1+k\linebreak\leqslant (j-1)N_1+l$, which is the order we have defined in
the previous section.
Let $\alpha = (i,k),\beta = (j,l)$, then the decomposition can be written in a conciser form,
\begin{equation}
   \label{eq:1}
 \rho = \adj  X X,\; X =
  \begin{pmatrix}
    &&\\
    &S_{\alpha\beta}X_{\alpha}&\\
    &&
  \end{pmatrix}.
\end{equation}

Recall the definition of SPPT state in Ref.~\cite{SPPT3partite}:
\begin{defi}
  Let $\rho=\adj XX$ be a density matrix in the tripartite system $A_1:A_2:A_3$ with $X$ being  the form as \cref{eq:1}. Then
  $\rho$ is said to be  SPPT if
  \begin{equation}
    \label{eq:4}
    \rho^{\intercal_{12}} = \adj{Y} Y,
  \end{equation}
  where
\begin{equation}
  \label{eq:3}
  \begin{split}
    Y & =
    \begin{pmatrix}
      &&\\
      &\adj S_{\alpha\beta}X_{\alpha}&\\
      &&\\
    \end{pmatrix}.
  \end{split}
\end{equation}
\end{defi}

Note that in the above definition, the condition~\eqref{eq:4} is equivalent to
\begin{equation}
  \label{eq:5}
  \sum_{\alpha\leqslant \gamma_1} \adj{X_{\alpha}}\left[S_{\alpha,\beta},\adj{S_{\alpha,\beta'}}\right] X_{\alpha}=0,
\end{equation}
where  $\gamma_0 = (1,1)$, $\gamma_1  = (N_1,N_2)$,  and $\gamma_0 \leqslant\beta\leqslant\beta'\leqslant \gamma_1$.

Similarly, super SPPT (SSPPT) are also defined for tripartite system in Ref.~\cite{SPPT3partite}.
\begin{defi}
  Let $\rho$ be the SPPT state with a decomposition of form as \cref{eq:1}, then $\rho$ is SSPPT if
  \begin{equation}
    [S_{\alpha,\beta},\adj S_{\alpha,\beta'}] = 0, \; \gamma_0\leqslant\alpha\leqslant \beta \leqslant \beta' \leqslant \gamma_1.
  \end{equation}
\end{defi}

Note that $\alpha$ refers to the $n(\alpha)$-th row and $\beta$ refers to the $n(\beta)$-th column. Therefore, if we
reorder the 2-tuple $\alpha,\beta$ by a single index $n(\alpha)$ and $n(\beta)$ respectively, then the definition of
SPPT will be identical to that in the bipartite system. This implies that  some  properties  of SPPT states may hold only under the
bi-partition $A_1A_2:A_3$. 

 PPT, as is well-known,  is one of the most important features for SPPT states. However, the tripartite SPPT states
 defined here may lose this property. Here we construct an example show this defect.

Let
\begin{equation*}
 \rho = \adj v v,\; v = (1,0,0,0,0,0,1,0)
\end{equation*}
be the density matrix in $2\otimes  2\otimes 2$ system and
\begin{equation}
  \label{multiSPPTv0.6:eq:14}
  X_{11} =
  \begin{pmatrix}
    1&0\\
    0&0
  \end{pmatrix}.
\end{equation}
Then $\rho$ can be written as
\begin{equation}
  \label{multiSPPTv0.6:eq:15}
  \rho = \adj X X,X =
  \begin{pmatrix}
    X_{11} & 0 & 0 & X_{11}\\
    0&0&0&0\\
    0&0&0&0\\
    0&0&0&0
  \end{pmatrix}.
\end{equation}
Put
\begin{equation}
  \label{multiSPPTv0.6:eq:16}
  S_{\alpha,\beta} =
  \begin{cases}
    \I , &\alpha=\beta = (1,1),\\
    0,&\text{otherwise}.
  \end{cases}
\end{equation}
Therefore $\rho$ satisfies the condition of SSPPT.\@ However, $\rho$ is
not a separable state. Since any pure state is separable if and only if it is PPT, $\rho$ is consequently  not a PPT
state as well. This example illustrates that SPPT cannot guarantee PPT with this definition. Moreover, unlike the
bipartite case, super and pure SPPT states may not be separable.

Recall the  Theorem 1 of Ref.~\cite{SPPT3partite}:
\begin{thm}
  If $\rho$ us a super SPPT state in $N_1\otimes N_2\otimes N_3$ systems, then $\rho$ is bi-separable.
\end{thm}
It  pointed out in the proof that  tripartite SSPPT state is bi-separable with respect
to any of the  bi-partitions $A_1A_2:A_3$, $A_1:A_2A_3$ and $A_1A_3:A_2$. However this may not be true. Again, consider the example $\rho = \adj v
v$. It is  super SPPT by definition. But it is not separable under the bi-partition  $A_1:A_2A_3$ and $A_1A_3:A_2$.

In the next section, we will define our SPPT
and SSPPT states in another way. It turns out  many good properties will be preserved. 

\section{SPPT states in  multipartite case}%
\label{sec:mySPPT}%
In this section, we provide a new idea to define the tripartite SPPT states. We begin with the simplest case $2\otimes 2
\otimes N$, then we extend the idea to general $N_1\otimes N_2\otimes N_3$ tripartite system. Lastly, we give the
definition of SPPT in the arbitrary multipartite system. Correspondingly, the SSPPT states are also defined, which turn
out to be separable. In addition, we give some examples of SPPT states, which may be helpful to  shed new lights on
understanding the structure of PPT states in multipartite system.
\subsection{SPPT states in \texorpdfstring{$2\otimes 2 \otimes N$}{TEXT} system}
We begin with considering the simplest case when $\rho$ is a density matrix in $2\otimes 2 \otimes N$ system.

Let
\begin{equation}
  \label{eq:15}
  X =
  \begin{pmatrix}
    X_1 & SX_1\\
    0 & X_2
  \end{pmatrix}, Z =
  \begin{pmatrix}
    X_1 &\adj SX_1\\
    0 & X_2
  \end{pmatrix},
\end{equation}
where
\begin{align*}
  X_1 =
  \begin{pmatrix}
    Y_{11} & T_1 Y_{11}\\
    0 & Y_{12}
  \end{pmatrix}, X_2 =
        \begin{pmatrix}
          Y_{21} & T_{2} Y_{21}\\
          0 & Y_{22}
        \end{pmatrix}.
\end{align*}
Here $S$ is a diagonal block matrix
\begin{equation*}
  S=
  \begin{pmatrix}
    S_1 & 0\\
    0& S_2
  \end{pmatrix}.
\end{equation*}
Hence, $X$ can  be written as a $4\times 4$ block matrix with each block being a $N\times N$ matrix
\begin{equation}
  \label{eq:16}
  X =
  \begin{pmatrix}
    Y_{11} & T_1 Y_{11} & S_1Y_{11} & S_1T_1 Y_{11}\\
    0 & Y_{12} & 0 & S_2Y_{12}\\
    0 & 0& Y_{21} & T_2 Y_{21}\\
    0 & 0 & 0& Y_{22}
  \end{pmatrix}.
\end{equation}
Let
\begin{equation*}
  W =
  \begin{pmatrix}
    Y_{11} & \adj T_1 Y_{11} & \adj S_1Y_{11} & \adj{(S_1T_1)} Y_{11}\\
    0 & Y_{12} & 0 & \adj S_2Y_{12}\\
    0 & 0& Y_{21} & \adj T_2 Y_{21}\\
    0 & 0 & 0& Y_{22}
  \end{pmatrix}.
\end{equation*}

Now we are ready to define  SPPT states in the $2\otimes 2\otimes N$ system with the above notations.
\begin{defi}
  Let $\rho = \adj X X$ be a density matrix in  $2\otimes 2\linebreak\otimes N $ system where $X$ has the  form as \cref{eq:16}. Then
  $\rho$ is said to be SPPT if 
  \begin{equation}
    \label{eq:17}
    \rho^{\intercal_1} = \adj Z Z ,
  \end{equation}
  and
  \begin{equation}
    \label{eq:18}
    \rho^{\intercal_{12}} = \adj W W.
  \end{equation}
\end{defi}
Alternatively, the above two conditions in the definition of SPPT can be reformulated as
\begin{equation}
  \label{eq:19}
  \left\lbrace
    \begin{aligned}
     & \adj Y_{11} [T_1,\adj T_1] Y_{11} &= 0,\\
      &\adj Y_{11} [ S_1,\adj S_1] Y_{11} & =0,\\
      &\adj Y_{11}[T_1,\adj S_1] Y_{11} & = 0,\\
      &\adj Y_{11} [ S_1,\adj T_1\adj S_1] Y_{11} &  = 0,\\
      &\adj Y_{11} [ T_1,\adj T_1\adj S_1] Y_{11} & = 0,\\
      &\adj Y_{11} [ S_1T_1,\adj{(S_1T_1)}] Y_{11} \\
      &\quad- \adj Y_{12}[ \adj S_2,S_2]Y_{12} + \adj Y_{21} [ \adj T_2,T_2] Y_{21}&=0,\\
      &\adj Y_{11} [ S_1,\adj S_1] T_1Y_{12} & = 0,\\
      &\adj Y_{12} \adj T_1[S_2,\adj S_2 ] T_1 Y_{12}& = 0.
    \end{aligned}
  \right.
\end{equation}

Note that conditions~\eqref{eq:17} and~\eqref{eq:18} guarantee a ``strong PPT property''. This is one of the most different aspects compared with definition in previous section. 
Similarly, we can define a subclass of SPPT states which satisfy condition~\eqref{eq:19} automatically.
\begin{defi}
  Suppose that $\rho$ is an SPPT state with the   decomposition of the form as \cref{eq:16}. Then $\rho$ is called SSPPT if
  \begin{equation}
    \label{eq:20p}
    \left\lbrace
      \begin{array}{rl}%
        [S_i,\adj S_i] & = 0,i=1,2,\\
        {[}T_i,\adj T_i] &= 0,i=1,2,\\
        {[}S_1,\adj T_1] & = 0.\\
      \end{array}
    \right.
  \end{equation}
\end{defi}
 In the following theorem, we show that SSPPT can guarantee the separability.
\begin{thm}
  All SSPPT states in the tripartite system $2\otimes 2\otimes N$ are separable.
\end{thm}
\begin{proof}
According to the definition,   $ T_1,T_2,S_1,S_2$ are normal and $T_1$ commutes with $S_1$. Therefore, we have the
following   diagonalizations
  \begin{equation}
    \label{eq:21p}
    \begin{split}
      S_1   & = U   \Sigma_{1}  \adj U,  T_1  = U   \Sigma_{2}  \adj U,\\
      S_2  & = V_1 \Lambda_{1} \adj V_1
      ,  T_2  = V_2 \Lambda_2 \adj V_2,\\
    \end{split}
  \end{equation}
  where $\Sigma_i,\Lambda_i$ are the diagonal matrices and $U$, $V_1$ and $V_2$ are
  all unitary matrices.  Then
  \begin{widetext}
  \begin{equation}
    \label{eq:22p}
    \begin{split}
      X &=
      \begin{pmatrix}
        U \adj U Y_{11} &  U \Sigma_2 \adj U Y_{11} & U \Sigma_1 \adj U  Y_{11} & U\Sigma_1\Sigma_2 \adj U Y_{11} \\
        0 &  V_1 \adj V_1 Y_{12} & 0 & V_1 \Lambda_1 \adj V_1 Y_{12}\\
        0 & 0 & V_2 \adj V_2 Y_{21} &  V_2 \Lambda_2\adj V_2Y_{21}\\
        0 & 0 & 0 & Y_{22}\\
      \end{pmatrix}                    \\
      & =
      \begin{pmatrix}
        U & 0 & 0 &0 \\
        0 & V_2 & 0 &0 \\
        0 & 0 & V_1 & 0 \\
        0 & 0 & 0 & \I_N\\
      \end{pmatrix}
      \begin{pmatrix}
        \tilde Y_{11} &   \Sigma_2\tilde Y_{11} & \Sigma_1 \tilde Y_{11} & \Sigma_1\Sigma_{2} \tilde Y_{11} \\
        0 &  \tilde Y_{12} & 0 &  \Lambda_1 \tilde Y_{12}\\
        0 & 0 &  \tilde Y_{21} &   \Lambda_2\tilde Y_{21}\\
        0 & 0 & 0 & Y_{22}\\
      \end{pmatrix}\\
      & = G \tilde X,
    \end{split}
  \end{equation}
 
  where
  \begin{equation}
    \label{eq:23p}
    \tilde Y_{11} =  \adj U Y_{11}, \tilde Y_{12} =  \adj V_1 Y_{12}, \tilde Y_{21} =  \adj V_2 Y_{21}.   
  \end{equation}
   \end{widetext}
  Let
  \begin{equation*}
    \begin{array}{rclcccccrl}
      C_1 & = &\left(\right. &\tilde Y_{11} &   \Sigma_2\tilde Y_{11} & \Sigma_1 \tilde Y_{11} & \Sigma_1\Sigma_{2} \tilde Y_{11}&\left.\right)&, \\
      C_2 & = &\left(\right. &  0 &  \tilde Y_{12} & 0 &  \Lambda_1 \tilde Y_{12}&\left.\right)&,\\
      C_3 & = & \left(\right. &  0 & 0 &  \tilde Y_{21} &   \Lambda_2\tilde Y_{21}&\left.\right)&,\\
      C_4 & = & \left( \right.& 0 & 0 & 0 & Y_{22}&\left.\right)&.\\
    \end{array}
  \end{equation*}
 Note that $\rho = \adj X X$ and $G$ is a unitary matrix,  then  we obtain,
  \begin{equation*}
    \rho = \sum_{i=1}^4\adj C_i C_i.
  \end{equation*}
  On the other hand,
  \begin{equation*}
    \begin{split}
      C_1 &= (\I,\Sigma_1)\otimes (\I,\Sigma_2) \otimes \tilde Y_{11},\\
      C_2 & = (\I,\Lambda_1)\otimes (\,0\;,\I\,) \otimes \tilde Y_{12},\\
      C_3 & = (\,0\;,\I\,)\otimes (\I,\Lambda_2) \otimes \tilde Y_{21},\\
      C_4 &= (\,0\;,\I\,)\otimes (\,0\;,\I\,) \otimes \tilde Y_{22}.
    \end{split}
  \end{equation*}
  It follows that each $\adj C_i C_i$ is separable, which implies the separability of $\rho$.
\end{proof}
  Note that this proof can also be served as a method to
  find the separability  decomposition of a  $2\otimes 2\otimes N$ SSPPT state.

  Here we give a example of SPPT states with our definition.
  \begin{example}
     It was proved that every PPT state  $\rho$ supported on $2\otimes 2\otimes N$ 
 with rank $N$ is separable and has the canonical form~\cite{Karnas2001}
 \begin{equation}
   \label{multiSPPTv0.6:eq:2}
   \rho = \sqrt{D}
   \begin{pmatrix}
     \I_{N} \\ \adj B \\\adj C \\\adj B \adj C
   \end{pmatrix}
   \begin{pmatrix}
     \I_{N} & B & C & CB
   \end{pmatrix}
\sqrt{D},
\end{equation}
where $B,C,D$ are operators in the third subsystem and $B,C$ are normal commuting matrices.
  \end{example}

It is easy to check that
this canonical form is SPPT in our definition. Forward, it is also an  SSPPT state.

\subsection{SPPT states in \texorpdfstring{$N_1\otimes N_2\otimes N_3$}{TEXT} tripartite system}
In this subsection we will extend the SPPT states to  general tripartite system $N_1\otimes N_2\otimes N_3$. The basic
idea is to  require $\rho$ being SPPT under the bi-partition $A_1:A_2A_3$ and $A_1A_2:A_3$ simultaneously.

Let $\rho$ be the density matrix with a   decomposition $ \rho= \adj X X$ in the tripartite system $N_1\otimes N_2\otimes
N_3$.

Under the bipartite partition $A_1:A_2A_3$, $X$ can be written as an $N_1\times N_1$ block matrix,
  
\begin{equation}
  \begin{split}
  X
  = 
  \begin{pmatrix}
    X_1 & S_{12}X_1&\cdots &S_{1N_1}X_1\\
    0 & X_2&\cdots  & S_{2N_1}X_2\\
    \vdots &\vdots &\ddots &\vdots\\
    0 & 0  &\cdots& X_{N_1}\\
  \end{pmatrix}
  = 
  \begin{pmatrix}
    & & & &\\
    & & & &\\
    & &\mathrel{\raisebox{1em}{$S_{ij}X_i$}}& &\\
    & & & &\\
  \end{pmatrix},
  \end{split}
\end{equation}
where
\begin{equation}
  \label{multiSPPTv0.6:eq:3}
  S_{ij} =
  \begin{cases}
    0, &  \text{ if } i>j,\\
    \I_{N_2N_3},& \text{ if } i=j.
  \end{cases}
\end{equation}
Similarly, $X_i$ can
be written as an $N_2\times N_2$ block matrix, 
\begin{equation*}
  \begin{split}
    X_i &=
    \begin{pmatrix}
      X_{i1} & \sijklt{i}{12}X_{i1} & \cdots &\sijklt{i}{1N_2}X_{i1}\\
      0 & X_{i2} & \cdots & \sijklt{i}{2N_2}X_{i2}\\
      \vdots & \vdots & \ddots & \vdots \\
      0 & 0 & \cdots & X_{iN_2}
    \end{pmatrix}
    \\
    &=
    \begin{pmatrix}
      &&&&\\
      &&&&\\
      &&\mathrel{\raisebox{1em}{$\sijklt{i}{kl}X_{ik}$}}&&\\
      &&&&\\
    \end{pmatrix}.  
  \end{split}
\end{equation*}
Here the superscript $2$ in the matrices $\sijklt{i}{kl}$ indicates the subsystem $A_2$ and
\begin{equation*}
  \sijklt{i}{kl} =
  \begin{cases}
    \I_{N_3},& \text{ if }  k=l;\\
    0, & \text{ if } k>l.
  \end{cases}
\end{equation*}
In order to be compatible with the SPPT structure in the bipartite system $A_1A_2:A_3$, we require  $S_{ij}$ being 
diagonal,
\begin{equation}
  \label{eq:7}
    S_{ij}  =
    \begin{pmatrix}
      \sijkl{ij}{1}& 0 & \cdots &0\\
      0 & \sijkl{ij}{2}&\cdots & 0\\
      \vdots & \vdots & \ddots & \vdots \\
      0& 0& \cdots & \sijkl{ij}{N_2}\\
          \end{pmatrix}.
        \end{equation}
 Hence
 \begin{equation}
   \begin{split}
     \rho &= \adj X X,\\
     X  &=
     \begin{pmatrix}
       &&\\
       &S_{ij}X_i&\\
       &&\\
     \end{pmatrix},\\
    S_{ij}X_i  &=
    \begin{pmatrix}
      &&\\
      &  \sijkl{ij}{k} \sijklt{i}{kl} X_{ik} &\\
      &&\\
    \end{pmatrix}.
    \end{split}
   \label{sppt}
 \end{equation}
Now we are ready to define the SPPT state in general tripartite system with the matrices introduced above.
\begin{defi}
  Let $\rho$ be the density matrix in the tripartite system  $A_1:A_2:A_3$. It has a    decomposition 
  of the form as \cref{sppt}. Then we call $\rho$  SPPT w.r.t.\ the tripartite \linebreak $A_1:A_2:A_3$ system (or  simply
  SPPT), if $\rho$ is SPPT  under the bi-partitions $A_1:A_2A_2$ and $A_1A_2:A_3$ simultaneously.
\end{defi}

Note that the conditions for SPPT are equivalent to the following explicit  matrix equations
\begin{equation}
  \label{eq:24}
  \begin{split}
    \sum_{i=1}^{N_1} \adj X_i [S_{ip},\adj S_{iq}]X_i & = 0,\\
    \sum_i^{N_1}\sum_k^{N_2} \adj X_{ik}   [ \sijkl{ij}{k}\sijklt{i}{kl},  \adj{(\sijkl{ij'}{k}\sijklt{i}{kl'}) }]
    X_{ik} &= 0 .
  \end{split}
\end{equation}

It is clear from the definition that  SPPT states defined here are indeed PPT, i.e. positive under any partial
transpose. As in previous subsection, we 
can also define a subclass of SPPT states with \cref{eq:24} satisfied.
\begin{defi}
 Let  $\rho$ be a state in tripartite system $N_1\otimes N_2 \otimes N_3$ with a   decomposition $\rho = \adj X
 X$  of the form as \cref{sppt}. Then $\rho$ is said to be SSPPT if 
  \begin{equation}
    \label{eq:25}
    \begin{split}
      [S_{ip},\adj S_{iq}] &= 0,\\
      [ \sijkl{ij}{k}\sijklt{i}{kl},  \adj{(\sijkl{ij'}{k}\sijklt{i}{kl'}) }] & = 0.
    \end{split}
  \end{equation}
\end{defi}

As in the bipartite system, we can prove SSPPT states are separable, which is a good property we want to keep.
   \begin{thm}
     Every  SSPPT state in tripartite system  is separable.
   \end{thm}
   \begin{proof}
     Let $\rho$ be an SPPT in the tripartite system $N_1\otimes N_2 \linebreak\otimes N_3$, which possesses a decomposition as \cref{sppt}.
     It follows from the condition~\eqref{eq:25} that  $S_{ij}$ and $S_{ij'}$ are commutable for any given $i$.

     In particular,  given $i,k$,
     \begin{equation}
       \label{eq:26}
       \sijkl{ij}{k} \commutes \sijkl{ij'}{k}.
     \end{equation}
     Note that if we put  $l=k$ and $j'=i$ in \cref{eq:25}, then for any given $i,k$ we have
     \begin{equation}
       \label{eq:27}
       \sijkl{ij}{k} \commutes \sijklt{i}{kl'}.
     \end{equation}
     In the similar way, let $j=i$ and $j'=i$  in \cref{eq:25},  we obtain,
     \begin{equation}
       \label{eq:28}
       \sijklt{i}{kl'} \commutes \sijklt{i}{kl}.
     \end{equation}
     Therefore we have a simultaneous diagonalizations,
     \begin{equation}
       \label{eq:29}
       \begin{split}
         \sijkl{ij}{k}& = U_{ik} \Lijkl{ijk}\adj U_{ik},\\
         \sijklt{i}{kl}& = U_{ik} \Lijklt{ikl} \adj U_{ik},
       \end{split}
     \end{equation}
     where $U_{ik}$ are unitary matrices and $\Lijkl{ijk},\Lijklt{ikl}$ are diagonal matrices with
     \begin{equation*}
       \begin{split}
         \Lijkl{ijk} & = \mathrm{diag} ( \lijkl{ijk}{1},\lijkl{ijk}{2},\ldots,\lijkl{ijk}{N_3}),\\
         \Lijklt{ikl} & = \mathrm{diag} ( \lijklt{ikl}{1}\,,\lijklt{ikl}{2},\ldots,\,\lijklt{ikl}{N_3}).
       \end{split}
     \end{equation*}
     Let
     \begin{equation}
       \label{eq:30}
       \begin{split}
       U & =
       \begin{pmatrix}
         U_1 & 0 & \cdots & 0\\
         0 & U_2 & \cdots & 0\\
         \vdots & \vdots &\ddots& \vdots \\
         0 & 0 & \cdots & U_{N_1}
       \end{pmatrix},\\
       \tilde X &=
       \begin{pmatrix}
         Y_{11} & Y_{12} & \cdots &Y_{1N_1 } \\
         Y_{21} & Y_{22} & \cdots & Y_{2 N_1} \\
         \vdots & \vdots & \ddots & \vdots \\
         Y_{N_1,1} & Y_{N_1,2} & \cdots & Y_{N_1, N_1}
       \end{pmatrix},
       \end{split}
     \end{equation}
     where
     \begin{align*}
       U_i &=
             \begin{pmatrix}
               \adj U_{i1} & 0 & \cdots & 0\\
               0 & \adj U_{i2} & \cdots & 0\\
               \vdots & \vdots &\ddots& \vdots \\
               0 & 0 & \cdots & \adj U_{i,N_2}
             \end{pmatrix},\\
       Y_{ij} &=
                \begin{pmatrix}
                  &&\\
                  &\Lambda_{ijk}^1\Lambda_{ikl}^2X_{ik}&\\
                  &&
                \end{pmatrix},\\
       \tilde X_{ik}& = \adj U_{ik}X_{ik}.
     \end{align*}
     Since $U$ is unitary and $X = U\tilde X$, we have $\rho= \adj{\tilde X} \tilde X$.

     Suppose $\tilde X_{ik} = {(a_{ik1},a_{ik2},\ldots,a_{ikN_3})}^\intercal$ where each $a_{ikl}$ is a row vectors in
     $\complex^{N_3}$ space. Now
     consider the \linebreak$n(i,k,p)$-th  row of $\tilde X$, which is denoted by
     $v_{ikp}$.
     Then we have
     \begin{equation}
       \label{eq:31}
       v_{ikp} = w_{ikp}\otimes a_{ikp},
     \end{equation}
     where
     \begin{align*}
       w_{ikp} &= (y_{ikp1},y_{ikp2},\ldots, y_{ikpN_2}),\\
       y_{ikpj}& = (\lijkl{ijk}{p}\lijklt{ik1}{p},\lijkl{ijk}{p}\lijklt{ik2}{p},\ldots,\lijkl{ijk}{p}\lijklt{ikN_2}{p})\\
               & = \lijkl{ijk}{p}(\lijklt{ik1}{p},\lijklt{ik2}{p},\ldots,\lijklt{ikN_2}{p}).
     \end{align*}
     It follows that each $v_{ikp}$ is a product vector,
     \begin{equation}
       \label{eq:32}
       \begin{split}
         v_{ikp} =& (\lijkl{i1k}{p},\ldots,\lijkl{iN_1k}{p})\\
         &\otimes
       (\lijklt{ik1}{p},\ldots,\lijklt{ikN_2}{p})\otimes a_{ikp}.
       \end{split}
     \end{equation}
    Therefore  $\rho$ is separable.
   \end{proof}
   This proof can also be utilized as  a method to find the separability decomposition  of  SSPPT states in tripartite
   system.
   Now we end this subsection by given some examples of tripartite SPPT states.
   \begin{example}
   Recall that a state $\rho$ on $N_1\otimes N_2$ is said to be a CQ state~\cite{Piani2008} if it has the form
   \begin{equation}
     \label{multiSPPTv0.5:eq:1}
     \rho = \sum_i^{N_1} p_i \ket{i}\bra{i} \otimes \rho_i^{A_2},
   \end{equation}
   Where $\rho_i^{A_2}$ are  density matrices in $A_2$ subsystem. 
   It was proved that any CQ state is in fact SSPPT state. Similarly, we construct  a
   class of SPPT states  in  tripartite system $A_1:A_2:A_3$ as follows,
   \begin{equation}
     \label{multiSPPTv0.5:eq:2}
     \rho = \sum_i^{N_1}\sum_{j}^{N_2} p_{ij} \ket{ij}\bra{ij}\otimes \rho_{ij}^{A_3},
   \end{equation}
   where $\rho_{ij}^{A_3}$ are density matrices in subsystem $A_3$. 
   This is in fact an SSPPT states with  $ \sijkl{ij}{k} =\delta_{ij}I$,  $\sijklt{i}{kl}=\delta_{kl}I$ and $X_{ik} = p_{ik}\rho_{ik}^{A_3}$.
 \end{example}
 \begin{example}
   Recall that every PPT state $\rho$ in $N_1\otimes N_2\linebreak\otimes N_3$ with $\rank{\braket{00|\rho}{00}}=\rank{\rho } =
   N_3$ can be transformed into the following canonical form  by using  a reversible local operator \cite{Wang2005a},
   \begin{equation}
     \label{multiSPPTv0.6:eq:4}
     \rho = \adj T T,
   \end{equation}
   where
   \begin{equation}
     \label{multiSPPTv0.6:eq:5}
     T = (\I_{N_{3}}, A_2,\ldots,A_{N_1})\otimes (\I_{N_{3}},B_2,\ldots,B_{N_2}) .
   \end{equation}
   Moreover $A_i,B_i$ are a set of normal commuting matrices.
    \end{example}
   Now we show   this canonical form  is actually an  SSPPT  state.

   Assume $A_1 =B_1= \I _{N_3}$ and $X_{ij}=\I_{N_3},\forall i,j$.
   Let $S_{1,j,1}^1 = A_j$ and $S_{1,1l}^2 = B_l$ and  all the other $S_{ij,k}^1,S_{i,kl}^2$ are zero matrices. Then
   $T$ coincides with $X$ in \cref{sppt}.  Since all the $S_{ij,k}^1$ and $S_{i,kl}^2$ are normal commuting, $\rho$ is  SSPPT.\@

   \subsection{SPPT states in multipartite system}\label{sec:sppt-multi-partite}

   In this subsection, we will finally give the definition of SPPT in $(d+1)$-particle  system
   $N_1\otimes N_2\otimes \cdots \otimes N_d\otimes N_0$. To begin with, we will fix some notations for representing
   matrices in the multipartite system.
   
   Let $\alpha_{n}=(i_1,i_2,\ldots,i_n)$ with  $i_k\in\{1,2,\ldots,N_k\},k \linebreak\in\{ 1,2,\ldots,d\}$. Similarly, let
   $\beta_n = (j_1,j_2,\ldots,j_n)$. 
   Note that  the indexes $i_n,j_n$ correspond to the 
    $n$-th subsystem. For simplicity, we write  $\alpha_d,\beta_d$  as $\alpha,\beta$. Hence we can represent the
    matrices in a  conciser form. For example,
   \begin{equation}
     \label{eq:33}
     \begin{split}
       X_{\alpha_n} &= X_{i_1,i_2,\ldots,i_n},\\
       S_{\alpha_n,j_m}& = S_{i_1,i_2,\ldots,i_m,j_m,i_{m+1},\ldots,i_n},m\leqslant n.
     \end{split}
   \end{equation}
   
   Hereafter in this subsection,   $(\alpha,\beta)$-th  entry of a matrix to represent the element in $n(\alpha)$-th
   row and $n(\beta)$-th column.
   
   Let $\rho = \adj X X$ be a density matrix in the  $N_1\otimes N_2\otimes \cdots\linebreak \otimes N_d\otimes N_0$
   system.
   Consider the following class of upper triangular block matrix $X$, whose elements are $N_0\times N_0$ matrices. 
   The  $(\alpha,\beta)$-th entry of $X$ is
   \begin{equation}
     \label{eq:34}
     \prod_{p=1}^d S^p_{\alpha,j_p} X_{\alpha},\quad S^p_{\alpha,j_p},\,X_{\alpha}\in \complex^{N_0\times N_0},
   \end{equation}
   where
   \begin{align*}
     S^p_{\alpha_n,j_p} & = \mathrm{diag} (
                          S^p_{\alpha_n,1,j_p},\ldots,S^p_{\alpha_n,N_{n+1},j_p}),\\
     X_{\alpha_n} & =
                    \begin{pmatrix*}[c]
                      &\\
                      & S^{n+1}_{\alpha_n,i_{n+1},j_{n+1}}X_{\alpha_n,i_{n+1}}\\
                      & & &
                    \end{pmatrix*}\\
                            &=
                            \begin{pmatrix*}[c]
                      &\\
                      & S^{n+1}_{\alpha_{n+1},j_{n+1}}X_{\alpha_{n+1}}\\
                      & & &
                    \end{pmatrix*},\\
     S^p_{\alpha_n,j_p} & =
                          \begin{cases}
                            1, & j_p = i_p,\\
                            0, & j_p < i_p.
                          \end{cases}
   \end{align*}

  \begin{defi}
 Let     $\rho$ be the density matrix in the $(d+1)$-body system $N_1\otimes \cdots\otimes N_d\otimes N_0$ with the
 decomposition $\rho = \adj X X$ of the form  as \cref{eq:34}.
Then $\rho$ is said to be SPPT if 
    \begin{equation}
      \label{eq:35}
      \sum_{\alpha_{n}}      \adj X_{\alpha_n}  \left[\;\prod_{p=1}^{n}S_{\alpha_n,j_p}^p,\adj{\left( \prod_{q=1}^n S^{q}_{\alpha_n,j'_q} \right)}\;\right]X_{\alpha_n} = 0
    \end{equation}
  for any $\beta_n= (j_1,j_2,\ldots,j_n)$,
    $\beta_n'= (j_1',j_2',\ldots,j_n')$ and $n=1,2,\ldots,d$.
  \end{defi}

  The following  theorem shows that this generalization of  SPPT  preserves the PPT property.
  \begin{thm}
    Any SPPT state is PPT.\@
  \end{thm}
  \begin{proof}
    Consider the density matrix $\rho$ in the $N_1\otimes\cdots N_d\otimes N_0$ system.
    Suppose $\rho = \adj{X} X$, where $X$ has the form as  \cref{eq:34}.
    Recall a special partial transpose defined previously,
    \begin{equation*}
      \Gamma_n = \intercal_{\{1,2,\ldots,n\}}, n = 1,2,\ldots,d.
    \end{equation*}
    To prove the PPT property of $\rho$, it suffice  to show that $\rho^{\Gamma_{n}} $ is positive for any $n$.

    Consider the state under the bi-partition $A_1A_2\ldots A_n:A_{n+1}\ldots A_dA_0$. Then $X$ can be regarded as a
    $r\times r $ block matrix, where $r = \prod_{k=1}^nN_k$.
    Given any $\alpha_n = (i_1,i_2,\ldots,i_n)$ and $\beta_n = (j_1,j_2,\ldots,j_n)$, the $(\alpha_n,\beta_{n})$-th
    entry of $X$ is
    \begin{equation*}
      \prod_{p=1}^n S^p_{\alpha_n,j_p} X_{\alpha_n}. 
    \end{equation*}
    Let $Y$ be the  matrix whose $(\alpha_n,\beta_n)$-th entry is
    \begin{equation*}
      \adj{ \left( \prod_{p=1}^n S^p_{\alpha_n,j_p}\right)} X_{\alpha_n}. 
    \end{equation*}
    According to the conditions~\eqref{eq:35}, we have  $\rho^{\Gamma_n} = \adj{Y}Y$, which completes the proof.
  \end{proof}

  In a similar way to  tripartite system, we  can define a special sub-class of SPPT states.
  \begin{defi}
    Let $\rho=\adj X X$ be an SPPT state where $X$ has the form as \cref{eq:34}.
    Then $\rho$ is said to be SSPPT if 
    \begin{equation}
      \label{eq:ssppt}
      \left[\;\prod_{p=1}^{n}S_{\alpha_n,j_p}^p,\adj{\left( \prod_{q=1}^n S^{q}_{\alpha_n,j'_q} \right)}\;\right] =0,
    \end{equation}
    for any $ \alpha_n,\beta_n,\beta_n'$ and $n=1,2,\ldots,d$.
  \end{defi}

  The following theorem shows that SSPPT guarantees the separability in an arbitrary multipartite system.
  \begin{thm}
    Any SSPPT state is separable.
  \end{thm}
  \begin{proof}
    Let $\rho$ be the density matrix in the $(d+1)$-particle system $N_1\otimes \cdots N_d\otimes N_0$. 
    Suppose $\rho = \adj{X}X$ with $X$ being of  the form as  \cref{eq:34}.

    Consider the condition \eqref{eq:ssppt}. Given $k,l$, choose $\alpha_n$ and $\beta_n$ such that
    \begin{equation*}
      \begin{split}
        j_p &= i_p ,p=1,\ldots,n,p\neq k;\\
        i_{q'}& = i_{q}, q = 1,\ldots, n, q\neq l.
      \end{split}
    \end{equation*}
    It follows that
    \begin{equation}
      \label{multiSPPTv0.6:eq:6}
      S_{\alpha_n,j_p}^{p}= S_{\alpha_n,j'_q}^{q} = \I, p\neq k,q\neq l,
    \end{equation}
    which  implies that
    \begin{equation}
      \label{multiSPPTv0.6:eq:7}
\left[         S^k_{\alpha_n,j_k} S^l_{\alpha_n,j'_l} \right] =0.
\end{equation}
Forward by the structure of $S_{\alpha_n,j_p}^p$, we have
\begin{equation}
  \label{multiSPPTv0.6:eq:8}
  \left[         S^k_{\alpha,j_k} S^l_{\alpha,j'_l} \right] =0.
\end{equation}
   
That is to say $\{S_{\alpha,j_p}^p\}$  is a set of normal commuting matrices for fixed $\alpha$.

    Hence we have the simultaneous diagonalizations
    \begin{equation*}
      S^p_{\alpha,j_p} =  U_{\alpha} \Lambda^p_{\alpha,j_p}\adj U_{\alpha},
    \end{equation*}
    where $U_{\alpha}$ is an $N_0\times N_0$ unitary matrix and $\Lambda^p_{\alpha,j_p}$ are diagonal matrices with
    \begin{equation*}
      \Lambda_{\alpha,j_p}^p = \mathrm{diag} (\lambda_{\alpha,1,j_p}^p,\lambda_{\alpha,2,j_p}^p,\ldots,\lambda_{\alpha,N_0,j_p}^p).
    \end{equation*}
    Let $U$ be the  block matrix whose $(\alpha,\alpha)$-th entry is $U_{\alpha}$ and other entries are zero. Put
    $\tilde X = U X$ and $\tilde X_{\alpha} = U_{\alpha}X_{\alpha}$, then $\rho = \adj{\tilde X }\tilde X$.

Let 
    \begin{equation*}
      \tilde X_{\alpha}  =
      \begin{pmatrix}
        a_{\alpha,1}\\
        a_{\alpha,2}\\
        \vdots\\
        a_{\alpha,N_0}
      \end{pmatrix},
    \end{equation*}
    where each $a_{\alpha,i_0}$ is a row vector in $\complex^{N_0}$. Note that
    $(\alpha,\beta)$-th entry of $\tilde X$  is
    \begin{equation*}
     \left( \prod_{p=1}^n \Lambda_{\alpha,j_p} \right)\tilde X_{\alpha}.
   \end{equation*}
   Hence we have,
    \begin{equation}
      \label{eq:38}
      v_{\alpha,i_0} = w_{\alpha,i_0} \otimes a_{\alpha,i_0},
    \end{equation}
    where
    \begin{equation*}
      \begin{split}
        w_{\alpha,i_0} &= \otimes _{p=1}^d y_{\alpha,i_0}^p,\\
        y_{\alpha,i_0}^p & = (\lambda_{\alpha,i_0,1}^p,\lambda^p_{\alpha,i_0,2},\ldots,\lambda^p_{\alpha,i_0,N_p})\in
        \complex^{N_p}.
      \end{split}
    \end{equation*}
    Now that each row of $\tilde X$ is a product vector, it follows that  $\rho$ is  separable.
  \end{proof}

   The following lemma shows an example of SPPT state in the  general multipartite system.
   \begin{lem}
     Any pure state is separable if and only if it is SPPT.\@
  \end{lem}
  \begin{proof}
    Since any pure PPT state is separable, it suffices to prove that pure product state is  indeed SPPT states.
    Let $\rho$ be a pure state in $N_1\otimes \cdots N_d\otimes N_0$ system. Hence $\rho$ can be written as 
    \begin{equation}
      \label{eq:2}
      \begin{split}
        \rho &= v \adj v,
      v = \adj w, \\
     w &=(\otimes_{i=1}^{d}w_i)\otimes w_0,\\
      w_p & = (w_{p,1},w_{p,2},\ldots,w_{p,N_p}),1\leqslant p\leqslant d.
      \end{split}
    \end{equation}
    Let $\alpha_1 = (1,1,\ldots,1) $ and $X_{\alpha_1}$ be a $N_0\times N_0$ matrix whose first row is
    $w_0$ and all other entries are zeros. Consider $w$ as a block vector with each block being a $N_0$ dimensional
    row vector, then the $\alpha$-entry of $w$ is $\prod_{p}^d w_{p,j_p} w_0$.

   Let $\beta  = (j_1,j_2,\ldots,j_d)$ and 
    \begin{equation}
      \label{eq:11}
      \begin{split}
        S_{\alpha_1,j_p}^p& =\mathrm{diag}( w_{p,j_p},0,\ldots,0).\\
      \end{split}
    \end{equation}
   For any other $\alpha\neq \alpha_1$, let $S_{\alpha,j_p}^p = 0$ and $X_{\alpha}=0$.  Then we can write $\rho$ as $\rho = \adj X X$, where
    the  $(\alpha,\beta)$-th entry of $X$ is
    \begin{equation}
      \label{eq:47}
        \prod_{p=1}^{d}S_{\alpha,j_p}^p X_{\alpha},      
      \end{equation}
      which has the same structure as in the definition of SPPT. Moreover, $S_{\alpha,p}^p$ here are all commuting
      normal matrices, hence it is  SSPPT.\@
    \end{proof}

    We end this subsection by giving another example in the multipartite system.
    \begin{example}
      It was proved that any PPT state  in $N_1 \otimes \cdots\linebreak\otimes  N_d\otimes N_0$ is separable~\cite{Wang2005a} if
      \begin{equation}
        \label{multiSPPTv0.6:eq:12}
        \begin{split}
          \rank{\rho} & = 
          \rank{ \bra{0_1,0_2,\ldots,0_d} \rho \ket{0_1,0_2,\ldots,0_d}} \\
          &= N_0.
          \end{split}
      \end{equation}
    It then has a canonical form by using a reversal local operator:
    \begin{equation}
      \label{multiSPPTv0.6:eq:13}
      \rho = \adj T  T,
    \end{equation}
    where
    \begin{align*}T =& (D_1^1,D_2^1,\ldots,D_{N_1}^1)\otimes (D_1^2,D_2^2,\ldots,D_{N_1}^2)\cdots \\
     & \otimes (D_1^d,D_2^d,\ldots,D_{N_d}^{d}).
      \end{align*}
    Here $D_1^i = \I$ and
    $D_p^{q}$ are a set of mutually commuting normal matrices.
  \end{example}
  Suppose that  $X$ has the form as \cref{eq:34}.
    Let $\alpha_0 = (1,1,\ldots,1)$ be a $d$-tuple.
    Put $S_{\alpha_0,j_p}^p = D_{j_p}^p$ for any $p$. And $S_{\alpha,j_p}^p = 0$ for all other $\alpha\neq
    \alpha_{0}$. And $X_{\alpha_0} = \I$. Simple calculation follows that $\rho = \adj X X$. Note that
    $S_{\alpha,j_p}^p$ are all mutually normal commuting, hence it is SSPPT.\@

    \section{Sufficient separability conditions of SPPT states}%
  \label{sec:suff-cond-sppt}

  In this section, we will consider the separability  conditions  for SPPT state.

  Let $\rho$ be a density matrix in $2\otimes d$ system with a block Cholesky decomposition,
  \begin{equation}
    \label{2d}
    \begin{split}
      \rho &= \adj{X}X,\\
      X &=
      \begin{pmatrix}
        X_1 & S X_1\\
        0 & X_2
      \end{pmatrix}.
    \end{split}
  \end{equation}
  Note that $\rho$ is SPPT if
  \begin{equation}
    \label{multiSPPTv0.6:eq:9}
    \adj X_1 (S_1\adj S_1-\adj S_1 S_1) X_1 = 0.
  \end{equation}
  It has been proved that SPPT states in $2\otimes 4$ system is separable~\cite{ha2013separability}. In fact we have the following conclusion.
  \begin{lem}\label{lem:2} Let $\rho$ be an SPPT state of the form as \cref{2d}. Then 
    $\rho$ is separable in either of the following cases
    \begin{enumerate}
      \item $d \leqslant 4$;
      \item $\rank{X_1} = d$.
    \end{enumerate}
  \end{lem}
  The second condition can be further improved as follows. 
  \begin{lem}
    Let $\rho$ be an SPPT state of the form as \cref{2d}. Then $\rho$ is separable if $\mathrm{Im}(S)\subset \mathrm{Im}(X_1)$
    or $\mathrm{Im}(\adj S)\subset \mathrm{Im}(X_1)$.
  \end{lem}
  \begin{proof}
    It suffices to prove  for the case $\mathrm{Im}(S)\subset\mathrm{Im}(X_1)$ since otherwise we can consider the
    partial transposed state $\rho^{\intercal_1}$.

   Suppose $\rank{X}< d$, then it has a SVD decomposition,
    \begin{equation}
      \label{eq:50}
      X_1 =  U \Lambda\adj V=
      U
      \begin{pmatrix}
        \Sigma&0\\
        0&0
      \end{pmatrix}
      \adj V ,
    \end{equation}
    where $U,V$ are unitary matrices and $\Sigma $ is a diagonal matrix with dimension less than $d$.
    Let
    \begin{equation}
      \label{multiSPPTv0.6:eq:10}
      \begin{split}
        \sigma &= \adj Y Y,\\
        Y &=
        \begin{pmatrix}
          \Gamma & \adj U S U \Gamma\\
          0 & X_2 V 
        \end{pmatrix}.
        \end{split}
      \end{equation}
      Then $\rho = (\I \otimes V)\sigma (\I \otimes \adj V)$. Simple calculation gives that $\sigma$ is also SPPT.\@
     Write $S$ in block matrix form according to that of  $\Lambda$,
    \begin{equation}
      \label{eq:52}
      S =
      \begin{pmatrix}
        S_1&S_2\\
        S_3&S_4
      \end{pmatrix}.
    \end{equation}
    
    Note that
    \[\mathrm{Im}(S)\subset \mathrm{Im}(X_1) \Leftrightarrow\mathrm{Im}(\adj US U)\subset
      \mathrm{Im}(\Lambda),\]
    it  follows that
    $S_3=0,S_4=0$. By  the condition~\eqref{multiSPPTv0.6:eq:9}, we have $S_4=0$ and
    $S_1$ is normal. Therefore $S$  is normal, which implies  $\rho$ is separable.
  \end{proof}
  Another sufficient condition for separability of $2\otimes d $ SPPT state was given in Ref.~\cite{Guo2012}.
  \begin{lem}%
  \label{AgreaterD}
    Let \begin{equation} \rho =
      \begin{pmatrix}
        A & B\\
        \adj B & D
      \end{pmatrix}
      \label{abd}
    \end{equation}
    be a density matrix in the $2\otimes d $ system.  If $A>D$, then $\rho$ is SSPPT and thus separable.
  \end{lem}
  Noted that, when $\rho$ is written in \cref{2d}, the sufficient condition in the above lemma is equivalent to
  \begin{equation*}
    \adj{X_1}X_1 > \adj{X_1}\adj{S}SX_1 +\adj X_2X_2.
  \end{equation*}
  We can further relax the condition by
  \begin{lem}
    Let $\rho=\adj XX$ with $X$ being of the form as \cref{2d} in $2\otimes d$ system. Then $\rho$ is separable if
    \[\adj X_1 X_{1}> \adj X_1\adj S SX_1.\]
  \end{lem}
  \begin{proof}
    Now $\rho$ can be written as
    \begin{equation*}
      \begin{split}
        \rho & =
        \begin{pmatrix}
          \adj X_1 \adj SSX_1 & \adj X_1 S X_1\\
          \adj X_1 \adj S X_1 &            \adj X_1 \adj SSX_1 \\
        \end{pmatrix}
        \\
        &\quad+
        \begin{pmatrix}
          X_1\adj X_1 -     \adj X_1 \adj SSX_1 &0\\
          0 & \adj X_2 X_2\end{pmatrix}.
      \end{split}
    \end{equation*}
    The former term is a positive Toeplitz block matrix which is separable by the Proposition 1 in
    Ref.~\cite{Gurvits2002}. Since the latter term is separable, $\rho$ is separable.
  \end{proof}
  To sum up the conditions in term of $S$, we have
  \begin{corollary}
    Let $\rho$ be an SPPT state of the form as \cref{2d}. Then $\rho$ is separable if $S$ is in any of the following casesa
    \begin{enumerate}
      \item $S$ is contractive,           
      \item $S$ is normal,
      \item $\mathrm{Im}(S)\subset \mathrm{Im}(X_1)$,
      \item Dimension of $S$ is  less than or equal to $4$.
    \end{enumerate}
  \end{corollary}

  Kil-Chan Ha  constructed a  $2\otimes 5$ SPPT state which is entangled~\cite{Ha2012}. Here we study further about
  $2\otimes 5$ SPPT states. Before that we recall the definition of edge state.
  \begin{defi}
    Let  $\sigma$ be multipartite state. It is said to be an edge state if there does not exist $\ket{x,y}$ such that
  \begin{equation}
    \label{multiSPPTv0.6:eq:11}
    \begin{split}
    \ket{x,y} &\in \range{\sigma},\\
    \ket{\conj x ,y}&\in \range{\sigma^{\intercal_1}}.
    \end{split}
  \end{equation}
  \end{defi}
  
  \begin{thm}
    $\rho$ is an SPPT state in $2\otimes 5$ system of the form as \cref{2d}, then $\rho$ is separable except the following case:
    \begin{equation}
      \label{eq:53}
      \begin{cases}
         &\rank{X_1}=4,\\
          &\rank{\sigma} =\rank{\sigma^{\intercal_1}} =5,\\
          &\sigma \text{ is an edge state},
      \end{cases}
    \end{equation}
    where
    \begin{equation*}
      \begin{split}
        \sigma & = \adj{W}W,\quad
        W  =
        \begin{pmatrix}
          X_1 & SX_1\\
        \end{pmatrix}.
      \end{split}
    \end{equation*}
  \end{thm}
 
  \begin{proof}

    By \cref{lem:2}, $\rho$ is separable if it has full rank. Hence  we assume  $r=\rank{X_1}\leqslant 4$.
    Consider $\sigma$,  an SPPT state  supported in $2\otimes r$ subspace. If $r<4$, then it is separable by
    the Peres-Horodecki criterion. Therefore, $\rho$ is separable. We  are thus able to  assume $r=4$.

    By the PPT property, $\rank{\sigma}\geqslant r = 4$. And it is separable if $\rank{\sigma} = 4$ or
    $\rank{\sigma^{\intercal_1}}=4$. Hence we only need to consider the case when 
    $\rank{\sigma}=\rank{\sigma^{\intercal_1}}=5$.

    Since any $2\otimes 4$ birank$(5,5)$ state is entangled if and only if it is an edge state. Hence $\rho$ is separable
    when $\sigma$ is not an edge state, which completes our proof.
  \end{proof}

  Recall the $2\otimes 5$ SPPT entangled state in Ha's paper~\cite{ha2013separability},
  \begin{equation}
    \label{eq:14}
    X_1 =
    \begin{pmatrix}
      \I_4&0\\
      0&0
    \end{pmatrix}, S =
    \begin{pmatrix}
      0 & 1 & 0 & 0 & \beta_1\\
      0 & 0& 1& 0&0\\
      0 & 0& 0& 1 &0\\
      0 & 0& 0& 0& \beta_2\\
      \beta_2 & 0 & 0&\beta_1 & 0
    \end{pmatrix},
  \end{equation}
  where $\beta_1= \sqrt{(1-b)/2b}$ and $\beta_2 = \sqrt{ (1+b)/2b}$ with $0< b<1$. Put $X_2 = 0$.  Then the defined
  $\sigma $ is
  \begin{equation}
    \sigma = 
    \left(
    \begin{array}{ccccc|ccccc}
      1 & 0 & 0 & 0 &0   & 0 & 1 &0 & 0&0\\
      0 & 1 & 0 & 0 &0   & 0 & 0 & 1 &0&0\\
      0 & 0 & 1 & 0 &0   & 0 & 0& 0 & 1&0\\
      0 & 0 & 0 & 1 &0   & 0 & 0& 0 & 0&0\\
      0 & 0 & 0 & 0 &0   & 0 & 0& 0 & 0&0\\\midrule
      0 & 0 & 0 & 0 &0   & \gamma_1 & 0& 0 & \gamma_2&0\\
      1 & 0 & 0 & 0 &0   & 0 & 1& 0 & 0&0\\
      0 & 1 & 0 & 0 &0   & 0 & 0& 1 & 0&0\\
      0 & 0 & 1 & 0 &0   & \gamma_2 & 0& 0 & \gamma_1&0\\
      0 & 0 & 0 & 0 &0   & 0 & 0& 0 & 0&0\\
    \end{array}
    \right),
  \end{equation}
  where $\gamma_1 = (b+1)/2b, \gamma_2 = \sqrt{b^2-1}$. $\sigma $ is supported on $2\otimes 4$ subspace and is of birank
  $(5,5)$ state. By computing all the product vectors in range $\sigma$ and $\sigma^{\intercal_1}$ respectively, it
  follows that $\sigma$ is an edge state, which coincides with our theorem.

  Furthermore, we have studied the rank 4 SPPT state.
  
\begin{thm}
Any SPPT state of rank less than or equal to 4 is separable.
\end{thm}
\begin{proof}
  Since all the rank $1,2$, and $3$ PPT states are separable, it only  remains to consider the  rank $4$ states.
  It has been proved in Refs.~\cite{Chen2011,Chen2013a} that  any rank four PPT state is   separable  
  except in the $2\otimes 2\otimes 2$ and $3\otimes 3$ systems. And for these systems, the state is separable if and only
  if its range  contains a product vector.
  Therefore, it suffices to prove that $\range{\rho}$ contains a product vector in 
  $2\otimes 2 \otimes 2$ and $3\otimes 3$ systems respectively. 
  
  Firstly, we consider the $2\otimes 2\otimes 2$ case.  Let $\rho=\adj X X$ be an SPPT state in $2\otimes 2 \otimes 2$
  system, where $X$ satisfies conditions  \eqref{eq:16} and \eqref{eq:19}.  Note
  that $\rho$ has a product vector in its range is equivalent to that $X$ has a product vector in its row range.

  If $Y_{22}\neq 0$, then $\range{X^{\intercal}}$ contains a  product vector.

  Let $\rho'=\adj{(Y_{21},T_2Y_{21}) }(Y_{21},T_2Y_{21})$, then $\rho'$ is a $2\otimes 2$ state. It is known from
  Ref.~\cite{Sanpera1997} that any two dimensional subspace of $2\otimes 2$ system  always contains a product
  vector. If $\rank{Y_{21},T_2Y_{21}} = 2$, then $\rho'$ contains a product vector in its range, namely $u$. Moreover,
  $(0,1)\otimes u$ is a product  vector in the range of $\rho$. On the other hand, if
  $\rank{Y_{21},T_2Y_{21}} = 1$, we claim that $\rho$ also contains a product vector in its range. Consider  the SVD
  decomposition of $Y_{21}$, denoted by $Y_{21} = U\Sigma \adj V$. Then the
  $5,6$-th rows of $X$ is
  \begin{equation*}
    U
    \begin{pmatrix}
     0&0& 0 & 0 & \sigma & 0 & t_1\sigma & 0 \\
     0&0& 0 & 0 & 0      & 0 & t_2\sigma & 0 \\
    \end{pmatrix}\adj V,
  \end{equation*}
  where
  \begin{equation*}
    \Sigma =
    \begin{pmatrix}
      \sigma & 0 \\
      0 & 0
    \end{pmatrix}, \adj U T_2 U =
    \begin{pmatrix}
      t_1& t_3\\
      t_2 & t_4
    \end{pmatrix}.
  \end{equation*}
  Note that  $( 0 , 0,0,0 , \sigma, 0 , t_1\sigma , 0)$ is always a product vector for any $t_1$.
  
  Similar way, we can show that $\rho$  contains a product vector in its range if $Y_{12}\neq 0$.

  However, $\rank{\rho}=4$ contradicts with  $Y_{12}=0,Y_{21}=0, \text{ and } Y_{22}=0$. It follows than $\rho$ is separable.

  Next, we consider the case when $\rho$ is in $3\otimes 3$ bipartite system.  Let $\rho$ be an SPPT state with
  \begin{equation*}
    \begin{split}
      \rho& = \adj{X}{X},\\
      X & =
      \begin{pmatrix}
        X_1 & S_{12} X_1 & S_{13}X_1 \\
        0 &X_2 &S_{23 }X_2\\
        0 & 0 & X_3
      \end{pmatrix}.
    \end{split}
  \end{equation*}
  If $X_3\neq 0 $, then $\rho$ contains a product in its range.

  Now assume  $X_3 =0$. If $X_1 $ has full rank, then $(X_2,S_{23}X_2)$ must be  rank one.  Suppose
  \begin{equation*}
    X_2 =
    \begin{pmatrix}
      \lambda_1 a\\
      \lambda_2 a\\
      \lambda_3 a
    \end{pmatrix},
  \end{equation*}
  where $a$ is a row vector in $\complex^3$. The row range of $S_{23}X_2$ will be contained in that of $X_2$, 
  we have
  \begin{equation*}
    S_{23}X_2 =
    \begin{pmatrix}
      \sigma_1 a\\
      \sigma_2 a\\
      \sigma_3 a\\
    \end{pmatrix},
  \end{equation*}
  for some $\sigma_i$.
  
  However  $(0,\lambda_i,\sigma_i)\otimes a$ is a product vector in the row range of $\rho$.

  Therefore, it only remains to consider the case when $\rank{X_1} = 2$.

  Since the rank of $\rho$ is 4, then $1\leqslant \rank{X_2}\leqslant
  2$. It has been proved that 
  $\rank{X_2}=1$ implies the separability of $\rho$, we  need only  to  consider the case when $\rank{X_2} = 2$.

  Let $\sigma = \adj{( 0,X_2,S_{23}X_2)}{(0,X_2,S_{23}X_2)}$ which is supported in $2\otimes 2$ subspace. Since any two
  dimensional subspace in $2\otimes 2$ system contains at least one product vector,  $\rho$ must contain a product vector
  in its  range.

  Above all, we conclude that any rank 4 SPPT state is separable.
\end{proof}

\section{Conclusion}%
\label{conclusion}
We extend the concept of well-known SPPT states to the arbitrary n-body system.  We compare the difference between the
definition of SPPT in  Ref.~\cite{SPPT3partite} and ours.  It turns out that our states 
can inherit the structure of PPT and many good properties as those in the bipartite system. For example, any  SPPT states are
separable, pure states are separable if and only if  they are SPPT, and any SPPT state of rank 4 is separable. Besides, we
also give some  sufficient  conditions for separability of SPPT states. In particular, for the $2\otimes 5$ SPPT states,
we showed that 
most of the states are separable except a special subclass.  We hope our work will be helpful for investigating the
structure of 
multipartite PPT states.

\begin{acknowledgments}
   The author is supported by an NUS Research Scholarship.
The author would like to thank the support and supervision  from Professor Chu Delin. The author also gratefully acknowledges Mr. Cui Hanwen for
the checking of the final script. Finally, the author would specifically like to highlight the impetus behind this work by Ms. Liu Cuizhen.
\end{acknowledgments}

\end{document}